\documentclass[11pt,a4paper]{article}
\usepackage[a4paper,margin=1in]{geometry}
\usepackage{mathtools}
\usepackage{wrapfig}
\usepackage{amssymb}
\usepackage{enumitem}
\usepackage{graphicx}
\usepackage{xcolor}
\usepackage{soul}
\usepackage{todonotes}
\usepackage{float}
\usepackage{caption}
\usepackage{thmtools}
\usepackage{amsthm}

\usepackage[labelformat=simple]{subcaption}

\usepackage{pgf,tikz}
\usetikzlibrary{decorations.markings, arrows.meta}

\usepackage[linesnumbered,ruled,vlined,algo2e]{algorithm2e}

\usepackage{hyperref}

\DeclarePairedDelimiter\ceil{\lceil}{\rceil}

\newtheorem{theorem}{Theorem}
\newtheorem{corollary}{Corollary}
\newtheorem{lemma}{Lemma}
\newcommand{\sharpS}{%
  \tikz[baseline=0ex, line width=1pt, scale=0.12]{%
    \draw (2,2) -- (1,2) -- (1,0) -- (0,0);
  }%
}

\SetCommentSty{mycommfont}

\newcommand{\Side}{\sqrt{n}}                 
\newcommand{\Lab}{\Lambda}                  
\newcommand{\id}{\mathsf{id}}               
\newcommand{\cert}{\mathsf{c}}              
\newcommand{\sweeplen}{\lambda}             
\newcommand{\sync}{\mathcal{SYNC}}

\begin{document}

\title{\textbf{Beyond the Boundary: Optimal Grid Dispersion by Spread, Reverse, and Collect}}

\author{
Manish Kumar\\
\small Indian Institute of Technology Dhanbad, Dhanbad, India\\
\small\texttt{manish27@iitism.ac.in}
\and
Anisur Rahaman Molla\\
\small Indian Statistical Institute Kolkata, Kolkata, India\\
\small\texttt{molla@isical.ac.in}
\and
Debasish Pattanayak\\
\small Indian Institute of Technology Indore, Indore, India\\
\small\texttt{debasish@iiti.ac.in}
\and
Subhajit Pramanick\\
\small Indian Institute of Technology Delhi, Delhi, India\\
\small\texttt{suvo.iitg17@gmail.com}
}
\date{}
\maketitle

\begin{abstract}
\emph{Dispersion} requires a set of mobile agents to occupy distinct nodes of a graph such that each node contains at most one agent. We study deterministic dispersion of $k\le n$ agents on an $n$-node oriented square grid, where each node has globally consistent North, East, South, and West port labels, from arbitrary initial configurations. This includes both \emph{rooted} configurations, where all agents are co-located initially, and \emph{general} configurations, where agents are initially distributed arbitrarily across the nodes.  On oriented square grids, Kshemkalyani, Molla, and Sharma~\cite{KMS20} gave an $O(\min(k,\Side))$-round algorithm in the local-communication model, where only agents on the same node can exchange information, and an $O(\sqrt{k})$-round algorithm in the global-communication model, where any two agents can communicate regardless of position, both using $O(\log k)$ bits per agent.  The bound in the local communication model is optimal when $k=\Theta(n)$, but it remains open whether local communication alone can match the $\Omega(\sqrt{k})$ time lower bound when $k=o(n)$. 
We solve this by presenting \textsc{Spread--Reverse--Collect} (SRC), a deterministic algorithm that disperses any $k\le n$ agents in $O(\sqrt{k})$ synchronous rounds using $O(\log k)$ bits of persistent memory per agent, matching both time and memory lower bounds.  SRC is based on a geometric primitive, a \sharpS-shaped parallel sweep, together with a simple priority rule.  Unlike earlier grid algorithms that gather agents at the boundary and spread out from there, each SRC sweep starts from the agents' current position, and when it reaches the boundary the group simply reflects its sweep back inside and continues. SRC therefore needs no knowledge of the grid size and runs unchanged on an infinite oriented grid. 


\end{abstract}

\maketitle

\section{Introduction}\label{sec:intro}

\subsection{Motivation and Background}
Coordinating autonomous mobile agents is a fundamental problem in distributed computing and robotics.  One of the most basic coordination tasks is \emph{dispersion}, introduced by Augustine and Moses~Jr.~\cite{AM18}, in which $k$ agents initially occupy arbitrary nodes of a graph and  move such that each node contains at most one agent.  
It has applications in various areas, such as robotic coordination in warehouses~\cite{wurman2008coordinating}, sensor deployment, search-and-rescue, and load balancing, while also capturing natural processes, including territorial behaviour~\cite{brawley1977territorial} and cell migration~\cite{le1999neural}.

We study optimal dispersion on an $n$-node square grid from any initial placement: a \emph{rooted} configuration, with all agents co-located, or a \emph{general} configuration, with arbitrary multiplicities.  We use the standard mobile-agent model.  Agents have unique identifiers (IDs) and execute CCM cycles (Section~\ref{sec:model}). We consider the \emph{local communication model}, where only co-located agents can communicate.  Nodes are anonymous and have no computational power.  Agents know neither $n$, the diameter $D$, nor $k$.  We work in the \emph{synchronous} ($\sync$) setting, where all agents are activated together and time is measured in \emph{rounds}.  We evaluate deterministic algorithms by \emph{time} and \emph{memory} (persistent bits per agent).

\subsection{Related Work}
Dispersion is studied across many graph classes and computational models~\cite{flocchini2019}.

\paragraph{General graphs.}\sloppypar
 Augustine and Moses~Jr.~\cite{AM18} introduced the problem and proved an $\Omega(\log n)$ per-robot memory lower bound, later extended to $\Omega(\log(\max(k,\Delta)))$ bits for maximum degree $\Delta$~\cite{KMS19}.  On arbitrary graphs, every deterministic algorithm needs $\Omega(D)$ rounds, where $D$ is the diameter.  Local-communication algorithms were improved from non-optimal bounds~\cite{KA19,KMS19,SSKM20} to near-linear time~\cite{SudoSSNKM25}.  Recent work gives essentially tight general-graph bounds: Kshemkalyani et al. obtain optimal $O(k)$ time under synchrony and $O(k\log k)$ under asynchrony~\cite{KshemkalyaniKMPS25SPAA}, and a subsequent work gives optimal asynchronous dispersion~\cite{PattanayakKMS25DISC}.  These bounds are optimal on graphs such as paths, where an agent may have to cross the network, but structured topologies can admit sublinear time.

\paragraph{Fault tolerance and variants.}
Fault-tolerant dispersion is also studied extensively, including crash faults~\cite{PS021,PSM23,ChandKSM25DAM} and Byzantine variants~\cite{MollaMM21,MollaIPDPS21}.  Dynamic and constrained models include path-connected dynamic graphs~\cite{Saxena025}, time-varying graphs~\cite{SaxenaK025}, and distance-$k$ dispersion on dynamic rings~\cite{MondalGS25DistanceK}.  These works typically use gather-and-scatter strategies or trade time for resilience.

\paragraph{Grids.}
Kshemkalyani, Molla, and Sharma~\cite{KMS20} first studied dispersion on \emph{oriented} grids, where globally consistent port labels let an agent move in a fixed compass direction without any sense of its absolute position.  They give an $O(\min(k,\Side))$-round, $O(\log k)$-memory algorithm in the local model, and reduce the time to $O(\sqrt{k})$ in the global model.  The local bound is optimal only for dense instances $k=\Theta(n)$; for sparse instances $k=o(n)$ it is a $\Theta\!\bigl(\min(k,\Side)/\sqrt{k}\bigr)$ factor above the $\Omega(\sqrt{k})$ lower bound. 
Subsequent work has broadened the grid setting in two directions.
Banerjee, Kumar, and Molla~\cite{BKM24} moved to \emph{unoriented} grids, where an agent no longer has a globally consistent notion of North, East, South, and West, so it cannot simply agree with other agents on a common direction; they study dispersion there both in the fault-free case and in the presence of faulty agents. In a later work~\cite{BKM25}, they return to oriented grids and provide \emph{fault-tolerant} dispersion that remains optimal even when some agents crash, so the surviving agents still reach a dispersed configuration. Himani and Pandit~\cite{HimaniPandit26TCS} consider \emph{triangular} grids, where each node has six neighbours instead of four, and achieve efficient dispersion without agents knowing the grid size in advance. Himani, Gorain, and Pandit study dispersion in torus networks~\mbox{\cite{himani2026}}. 
This is a related setting without a boundary, but a finite torus identifies nodes periodically. Consequently, extending SRC from finite square grids to the infinite grid does not, by itself, extend it to a torus.

Brandt, Uitto, and Wattenhofer~\mbox{\cite{brandt2018tight}} prove that three semi-synchronous finite automata cannot explore the infinite oriented grid. Their result concerns the number of constant-memory agents needed for exploration, rather than the time to disperse a variable population. For a constructive comparison, Brandt, Portmann, and Uitto~\mbox{\cite{BPU20Exploration}} implement a stack using distances between a constant number of agents. Their multiplication, division, and route-following routines encode exploration progress in those distances. SRC instead stores bounded path counters in each agent and assigns different subgroups to different columns of a finite square. Its forward and backward passes combine local settlement, label adoption, and recruitment to redistribute agents and recompute the active population. Thus the relevant distinction is the role of the traversal and its stored state, rather than the mere use of coordinated motion, reversals, or an infinite grid.

We close this gap for non-faulty oriented grids in the local model: our algorithm disperses the agents in $O(\sqrt{k})$ rounds with $O(\log k)$ bits, matching the bound that~\cite{KMS20} achieved only with global communication, in which an agent is allowed to communicate with any other agent in the graph.

\subsection{Our Contribution and Approach}
The $\Omega(\sqrt{k})$ time lower bound for oriented square grids is already known; we restate it as Theorem~\ref{thm:lb} for completeness.  Our contribution is an algorithm matching that bound:

\begin{itemize}
  \item \textsc{Spread--Reverse--Collect} (SRC), built around one geometric primitive --- a \sharpS-shaped parallel sweep --- and one arbitration rule: the \emph{higher-labelled group wins} in settling, mergers, and boundary contacts (Section~\ref{sec:algo}). 
  
  \item A correctness and complexity analysis (Theorems \ref{thm:main} and \ref{thm:correct}) 
  shows that SRC terminates in $O(\sqrt{k})$ rounds with $O(\log k)$ bits per agent.  The time bound uses a \emph{hierarchical merging-and-doubling} argument:  a label's certificate (its recorded group-size value) at least doubles whenever its \emph{level} (its growth counter) increases while the costs of subgroup meetings are counted under the higher label's cycle. Thus level $\ell$ costs $O\!\bigl(2^{\ell/2}\bigr)$, and the geometric sum up to certificate $k$ is $O(\sqrt{k})$.
  \item SRC naturally extends to an infinite oriented grid (Section~\ref{sec:disc}) unlike other related work~\cite{BKM25, KMS20} since boundary handling need not be invoked unless a boundary is encountered.
\end{itemize}

A further point sets our approach apart from earlier grid algorithms. Previous dispersion algorithms on grids rely on the boundary as a reference: agents first gather at a corner or an edge of the grid and then spread out from there, which needs the grid size to be known or discovered and thus the complexity depends on the grid size, failing to extend to an infinite grid. Our sweep never uses the boundary this way. Each group lays out and sweeps its own square from its current position, and when a sweep runs into the boundary the group simply reflects the square back inside (the \emph{bounce} in Section~\ref{subsec:boundary}, also refer to Fig.~\ref{fig:bounce}) and keeps going. The boundary is thus handled as a local event, not as a fixed landmark the algorithm depends on. Because of this, the same algorithm runs unchanged on an infinite grid, where there is no boundary at all. We see this boundary-independent sweep as the main technical novelty of the paper, and it may be of independent interest for other mobile-agent tasks on grids.

\section{Model and Preliminaries}\label{sec:model}

\subsection{The Oriented Square Grid}
The environment is an \emph{oriented square grid} $G=(V,E)$ with $n$ nodes.  We assume, as is standard, that $n$ is a perfect square, so the grid has side length $\Side$ and
$V=\{0,1,\dots,\Side-1\}^2$,
with an edge between every pair of nodes at Manhattan distance $1$.  Each node has four ports labelled $1,2,3,4$, identifying its incident edges as \emph{North} ($+y$), \emph{East} ($+x$), \emph{South} ($-y$), and \emph{West} ($-x$).  The labelling is \emph{consistent}: cardinal directions carry the same port label at every node, so an agent can always move in a cardinal direction by following the port-labels.  Nodes are anonymous and have no memory or computation. Fig.~\ref{fig:sqgrid} shows a grid with $n=25$ nodes.

\begin{figure}[t]
\centering
\begin{tikzpicture}[
    scale=0.6,
    x=2cm, y=2cm,
    node_style/.style={circle, draw=blue!80, fill=blue!20, minimum size=5mm, inner sep=0pt},
    port_style/.style={circle, draw=red!80, fill=red!10, minimum size=2mm, inner sep=0pt},
    every label/.style={font=\tiny, text=red!80}
]
\def\s{5}
\def\portdist{0.33}
\foreach \x in {0,...,\numexpr\s-1} {
    \foreach \y in {0,...,\numexpr\s-1} {
        \ifnum\y<\numexpr\s-1
            \pgfmathtruncatemacro{\ny}{\y+1}
            \draw (\x,\y) -- (\x,\ny);
            \node[port_style, label=center:{\textbf{1}}] at (\x,\y+\portdist) {};
            \node[port_style, label=center:{\textbf{3}}] at (\x,\ny-\portdist) {};
        \fi
        \ifnum\x<\numexpr\s-1
            \pgfmathtruncatemacro{\nx}{\x+1}
            \draw (\x,\y) -- (\nx,\y);
            \node[port_style, label=center:{\textbf{2}}] at (\x+\portdist,\y) {};
            \node[port_style, label=center:{\textbf{4}}] at (\nx-\portdist,\y) {};
        \fi
    }
}
\foreach \x in {0,...,\numexpr\s-1} {
    \foreach \y in {0,...,\numexpr\s-1} {
        \node[node_style] at (\x,\y) {};
        \node at (\x,\y) {\tiny (\x,\y)};
    }
}
\end{tikzpicture}
\caption{An oriented $\Side\times\Side$ grid (here $n=25$) with consistent port labels $1$ (North), $2$ (East), $3$ (South), and $4$ (West). These port labels induce the same orientation at every node.}
\label{fig:sqgrid}
\end{figure}
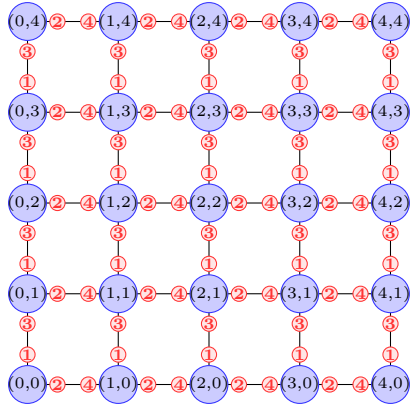

\subsection{Mobile Agents}
We are given a set $R=\{a_1,\dots,a_k\}$ of $k\le n$ autonomous mobile agents.  Each agent has a unique identifier of $O(\log k)$ bits and $O(\log k)$ bits of persistent memory. The identifier assumption bounds the representation length, rather than giving the number of agents ($k$) knowledge: identifiers can be encoded in $b=O(\log k)$ bits, so their size is at most $2^b=k^{O(1)}$. Identifiers break ties and are stored in labels. The system is synchronous: time proceeds in discrete \emph{rounds}, and in each round every agent executes one \emph{Communicate--Compute--Move} (CCM) cycle:
 \begin{itemize}
  \item \textit{Communicate.}  The agent reads the states of all agents currently at its node.
  \item \textit{Compute.}  From its own state and the states just read, it runs a deterministic local computation that outputs its next state and a move command: an outgoing port label, or the special command \emph{stay}.
  \item \textit{Move.}  The agent traverses the chosen edge, or remains in place.
\end{itemize}
Arbitrarily many agents may share a node, and arbitrarily many may traverse an edge in the same round.  Agents have no \emph{edge-crossing detection}: two agents that swap nodes across an edge in the same round do not notice one another.

\subsection{The Dispersion Problem}
Let $\mathcal{C}_0\colon R\to V$ be the initial configuration, possibly mapping several agents to one node. The \emph{dispersion problem} in this paper asks for a deterministic algorithm that, from any $\mathcal{C}_0$ with $k\le n$, reaches a configuration, where every node has at most one agent, in finite time and every agent stops.

\subsection{The Time Lower Bound}\label{subsec:lb}
The $\Omega(\sqrt{k})$ time lower bound follows from a simple ball-of-reach argument used in earlier grid work~\cite{KMS20}.  We restate it for completeness.

\begin{theorem}\label{thm:lb}
In general, any algorithm that disperses $k$ agents on an oriented square grid requires $\Omega(\sqrt{k})$ rounds in the worst case.
\end{theorem}

\begin{proof}
Consider the rooted configuration in which all $k$ agents are at a node $v$.  After $t$ rounds, an agent traverses at most $t$ edges, so it lies within Manhattan distance $t$ of $v$; hence every occupied node lies in the ball $B(v,t)=\{u:\lVert u-v\rVert_1\le t\}$. On a grid, this ball contains at most $2t^2+2t+1$ nodes (Fig.~\ref{fig:diamond}).  A dispersed configuration occupies $k$ distinct nodes, so we need $2t^2+2t+1\ge k$, i.e.,
\[
t \ge \tfrac12\bigl(\sqrt{2k-1}-1\bigr)=\Omega(\sqrt{k}).
\]

\begin{figure}[ht]
\centering
\resizebox{0.3\columnwidth}{!}{
\begin{tikzpicture}[
    x=1.6cm, y=1.6cm,
    node_style/.style={circle, fill=blue!10, minimum size=4mm, inner sep=0pt},
    near_style/.style={circle, line width=0.4mm, draw=black!80!green, fill=green!30, minimum size=4mm, inner sep=0pt}
]
\def\s{7}
\def\cx{3}
\def\cy{3}
\foreach \x in {0,...,\numexpr\s-1} {
  \foreach \y in {0,...,\numexpr\s-1} {
    \ifnum\y<\numexpr\s-1
      \pgfmathtruncatemacro{\ny}{\y+1}
      \draw (\x,\y) -- (\x,\ny);
    \fi
    \ifnum\x<\numexpr\s-1
      \pgfmathtruncatemacro{\nx}{\x+1}
      \draw (\x,\y) -- (\nx,\y);
    \fi
  }
}
\foreach \x in {0,...,\numexpr\s-1} {
  \foreach \y in {0,...,\numexpr\s-1} {
    \pgfmathtruncatemacro{\dist}{abs(\x-\cx) + abs(\y-\cy)}
    \ifnum\dist<3
      \node[near_style] at (\x,\y) {};
    \else
      \node[node_style] at (\x,\y) {};
    \fi
  }
}
\end{tikzpicture}
}
\caption{The ball of radius $t=2$ about the centre is a diamond of $2t^2+2t+1=2\cdot 2^2+2\cdot 2+1=13$ nodes. In general $\lvert B(v,t)\rvert\le 2t^2+2t+1$; requiring this to be at least $k$ gives $t=\Omega(\sqrt{k})$ (discussed in the proof of Theorem~\ref{thm:lb}).}
\label{fig:diamond}
\end{figure}
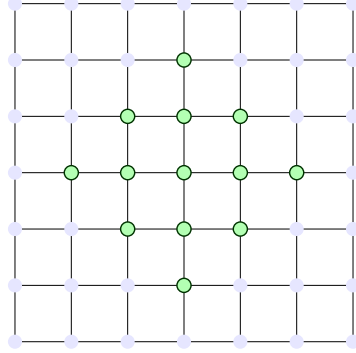
Therefore, at least one agent must move $\Omega(\sqrt{k})$ steps, and thus any algorithm requires $\Omega(\sqrt{k})$ rounds.
\end{proof}

\section{The Spread--Reverse--Collect (SRC) Algorithm}\label{sec:algo}

We first present the high-level idea of the algorithm (Algorithm \ref{alg:src}) and then move to a more technical description.

\subsection{Overview and Intuition}
\label{sec:overview}

We first describe the idea behind SRC at a high level, starting from the simplest case. Recall that the agents start on the nodes of the grid and must end with at most one agent per node. Agents that are co-located and move together form a \emph{group}. Every group carries a \emph{label} $\Lab=(\ell,\cert,\id)$, where $\ell$ is its \emph{level}, $\cert$ is its \emph{certified size}, and $\id$ is the identifier of its leader. The label determines the group's priority through lexicographic order. The level records how many times the group has successfully grown by a constant factor, while the certificate records the group size established at initialization or at the most recent level-up. An agent is \emph{active} while it participates in the current group movement and \emph{settled} once it has claimed a node as its final position. A settled agent remains on its node as a \emph{settled marker}: it does not move on its own, but it retains information about the sweep that placed it and can later be recruited by a same- or higher-labelled group.

The basic operation of SRC is a \emph{sweep}. Suppose a group has $m$ active agents at an anchor node $A$ (refer to Fig. \ref{fig:sweep}). It chooses a square of side $\gamma=2\ceil{\sqrt{m}}$ and partitions its agents into $\gamma$ subgroups. Each subgroup follows a $\sharpS$-shaped path and is responsible for one column of the square; an agent can settle only on a node covered by its own subgroup.

Since $\gamma^2\ge 4m$, intuitively a single sweep already solves the easiest case. Suppose there is exactly one group on the grid, with no other agents anywhere. Then every column its subgroups walk is empty, so each subgroup settles its agents one after another with room to spare, and the whole group is settled in this one sweep. This first sweep is the \textsc{Spread} phase, and when there is only one group, \textsc{Spread} alone finishes the task of dispersion. The difficulty arises when several groups sweep simultaneously and their square regions overlap. Whenever agents carrying different labels meet, the rule is simple: \emph{the higher label wins}. A lower-labelled moving agent that encounters a higher-labelled group adopts the higher label and follows the winner's current phase. However, the entire lower-labelled group need not be absorbed. Only the part that actually encounters the higher label is transferred; other subgroups of the lower-labelled group may continue independently or settle on free nodes. This \emph{partial absorption} may therefore leave a remaining active part of the lower-labelled group.

When a \textsc{Spread} of a group encounters occupied nodes, some agents may remain unsettled and gather at the far corner $C$ of the sweep; these remaining agents are called the \emph{overflow}. If the overflow is nonempty, the group performs two additional phases. In the \textsc{Reverse} phase, the overflow retraces the sweep from $C$ to $A$, absorbing moving agents of equal or lower label and claiming lower-labelled settled markers, but without settling new agents. Once the surviving moving agents are again co-located at $A$, they are re-ranked and re-partitioned into the same $\gamma$ subgroups. In \textsc{Collect}, they sweep from $A$ to $C$ once more and recruit settled or moving agents of equal or lower label that they encounter. Unlike \textsc{Reverse}, \textsc{Collect} can pick up already settled agents. It settles no agents itself; instead, it gathers the collected active agents at $C$. A \textsc{Spread} followed by \textsc{Reverse} and \textsc{Collect}, when needed, constitutes one \emph{cycle}.

At the end of a clean cycle, the group recomputes its active size at $C$. The certificate $\cert$ is not simply replaced by this new size. Instead, the group increases its level only when the newly collected active group has size at least twice its current certificate and the corresponding subgroup path has not encountered a higher label. In such a case, the new active size becomes the new certificate and the level increases by one. Thus, the certificate records a size that the group has previously established through genuine growth. This distinction is important when higher-labelled groups are encountered: agents may leave a lower-labelled group because the higher label wins, but the resulting decrease in active size does not decrease its certificate and cannot by itself justify a level-up.

The certificate-based growth rule controls the running time. Since the certificate at least doubles whenever the level increases, a group can undergo only $O(\log k)$ level-ups before its certificate reaches the total number $k$ of agents. The cost of a sweep depends on the current group size through $\gamma=2\ceil{\sqrt{m}}$. Since a group at level $\ell$ has certified size at least $2^\ell$, its sweep parameter is $\gamma=O(2^{\ell/2})$, and each sweep takes $O(2^{\ell/2})$ rounds. Summing these costs over all levels gives $O(\sqrt{k})$ rounds. 

Finally, a sweep may reach the boundary of the grid before it completes its intended path. In this case, the group \emph{bounces}: it changes its anchor to the boundary point, reverses its direction toward the interior, keeps the same label and $\gamma$, and starts another \textsc{Spread}. If the sweep is too large to fit within the grid, the group eventually reaches a corner and performs a terminating sweep that settles all remaining active agents. Since the boundary is detected locally through missing ports, the agents do not need to know the size of the grid.

Overall, SRC repeatedly combines three mechanisms: \textsc{Spread} disperses a group over a sufficiently large region, \textsc{Reverse} and \textsc{Collect} recover and enlarge a group when its spread is obstructed, and the higher-label-wins rule resolves interactions between competing groups. The level and certificate mechanism ensures that only genuine growth can increase the scale of future sweeps, while partial absorption prevents a lower-labelled remnant from gaining a higher level merely by changing its active population.

The rest of this section makes these ideas precise. 
Section~\ref{subsec:setup} fixes the labels, the sweep, the per-agent state, and the round-$0$ setup. Sections~\ref{subsec:spread}--\ref{subsec:collect} cover the three phases in turn. Section~\ref{subsec:interactions} gives the rules behind
``the higher label wins'', including how a lower label reacts to a higher one seen only through a settled agent. Section~\ref{subsec:boundary} handles the boundary and Section~\ref{subsec:decision} the end-of-cycle decision and the level-up rule.

\subsection{Labels, the Sweep, and Agent State}\label{subsec:setup}

\paragraph{\textbf{Labels.}}
Every agent carries a \emph{label} 
$\Lab=(\ell,\,\cert,\,\id)$. The first entry $\ell\in\mathbb{N}$ 
is the group's \emph{level}: it counts how many times the group has grown to at least twice its earlier size, each such step being a \emph{level-up} that raises $\ell$ by one (Section~\ref{subsec:decision}). The second entry $\cert\in\mathbb{N}$ is its \emph{certified size}: the number of agents in the group at its most recent level-up, or the number it started with at round~$0$ if it has had none. So the level counts the group's doublings and the certified size is the size it reached at the last of them; in particular, a group at level $\ell$ has certified size at least $2^{\ell}$. The third entry $\id$ is the largest identifier among the agents that formed this certified size, and we call that agent the \emph{leader}. The certified size is not the group's current number of active agents: that number changes during a cycle as agents settle or move to a higher label, whereas $\cert$ changes only at a level-up. In particular, $\cert$ never decreases; if a group loses agents to a higher label it keeps its certified size, and only a level-up can increase it. The label sets a group's priority through the \emph{lexicographic} order: $(\ell_1,\cert_1,\id_1)\prec(\ell_2,\cert_2,\id_2)$ if $\ell_1<\ell_2$, or $\ell_1=\ell_2$ and $\cert_1<\cert_2$, or $\ell_1=\ell_2,\ \cert_1=\cert_2$ and $\id_1<\id_2$. Since identifiers are unique, any two groups have different labels and can always be compared; $\preceq$ is the same order with equality allowed.

At the start of a cycle, the agents sitting together on one node learn their \emph{active size} $m$ (the number of active agents on the node) by local communication. For that cycle they set the \emph{sweep parameter} $\gamma = 2\ceil*{\sqrt{m}}$ and keep it fixed until the cycle ends. Since $\gamma\ge 2\sqrt{m}$, the square has area $\gamma^2\ge 4m$, at least four times the number of active agents, enough room for a group on its own to settle in a single \textsc{Spread} (Lemma~\ref{lem:threecycles}).

\paragraph{\textbf{The sweep.}}\label{subsec:sweep}
A sweep is fixed by three things: an \emph{anchor} $A=(a_x,a_y)$, a \emph{direction} $\delta=(d_x,d_y)\in\{-1,+1\}^2$, and the parameter $\gamma$.  Its \emph{ideal corner} is
\[
  C^\star=\bigl(a_x+(\gamma-1)d_x,\ a_y+(\gamma-1)d_y\bigr).
\]
If $C^\star$ is a real node of the grid ($C^\star\in V$), the sweep is \emph{clean}: its destination is $C=C^\star$ and its length is $\sweeplen=2(\gamma-1)$.  
If $C^\star$ lies outside the grid, the part of the square beyond the boundary is cut off, and we call the sweep \emph{clipped}. In this case the destination is obtained by bringing each coordinate of $C^\star$ to the nearest valid grid coordinate, which we write using the function $\operatorname{clip}$. For a value $z$ and a range with smallest allowed value $a$ and largest allowed value $b$ (so $a \le b$),
\[
  \operatorname{clip}(z,a,b)=\max\{a,\min\{b,z\}\}
  =\begin{cases}
      z, & \text{if } a\le z\le b,\\
      a, & \text{if } z<a,\\
      b, & \text{if } z>b,
   \end{cases}
\]
that is, $\operatorname{clip}(z,a,b)$ returns $z$ itself when $z$ already lies in the range $[a,b]$, and otherwise the endpoint of the range closer to $z$. In our use, $z$ is one coordinate of the ideal corner $C^\star$, and the range is the set of valid coordinates along that axis of the grid, from $a=0$ up to $b=\Side-1$. Applying this to both coordinates,
\[
  C=\bigl(\operatorname{clip}(C^\star_x,0,\Side-1),\ \operatorname{clip}(C^\star_y,0,\Side-1)\bigr),
\]
so a coordinate of $C^\star$ that already lies on the grid is kept unchanged, while one that points past the boundary is replaced by the boundary value on that side, $0$ or $\Side-1$. The \emph{boundary excesses} $t_x=\lvert C^\star_x-C_x\rvert$ and $t_y=\lvert C^\star_y-C_y\rvert$ then say how far $C^\star$ fell outside the grid along each axis, and the sweep length is $\sweeplen=2(\gamma-1)-t_x-t_y$. Agents do not know $\Side$; they only find out about a boundary when a port they need is missing. Moreover, the anchor $A$, the ideal corner $C^\star$, and the destination $C$ are names used only in the description and in the analysis. An agent never learns an absolute grid coordinate. It stores these points only as displacements from the node it currently occupies, and this relative storage is what keeps each agent's memory at $O(\log k)$ bits (Theorem \ref{thm:main}).

\begin{figure}[t]
\centering
\resizebox{0.4\columnwidth}{!}{
\begin{tikzpicture}[
    scale=0.6,
    node/.style={circle, fill=gray!50, inner sep=1.5pt},
    path_node/.style={circle, fill=blue, inner sep=2pt},
    start_node/.style={circle, fill=green!70!black, draw, thick, inner sep=2.5pt},
    end_node/.style={circle, fill=red!80!black, draw, thick, inner sep=2.5pt},
    label_style/.style={font=\small\bfseries}
]
\def\gridsize{6}
\pgfmathtruncatemacro{\N}{\gridsize-1}
\draw[step=1cm, color=gray!30, very thin] (0,0) grid (\N,\N);
\foreach \x in {0,...,\N} {
    \foreach \y in {0,...,\N} {
        \node[node] at (\x,\y) {};
    }
}
\def\startx{0}
\def\starty{0}
\def\sub{3}
\def\side{5}
\coordinate (A) at (\startx,\starty);
\coordinate (B) at (\startx+\sub,\starty);
\coordinate (D) at (\startx+\sub,\starty+\side);
\coordinate (C) at (\startx+\side,\starty+\side);
\draw[blue, ultra thick, line cap=round, line join=round,
      decoration={markings, mark=at position 0.6 with {\arrow{Stealth[length=3mm]}}},
      postaction={decorate}] (A) -- (B);
\draw[blue, ultra thick, line cap=round, line join=round,
      decoration={markings, mark=at position 0.55 with {\arrow{Stealth[length=3mm]}}},
      postaction={decorate}] (B) -- (D);
\draw[blue, ultra thick, line cap=round, line join=round,
      decoration={markings, mark=at position 0.4 with {\arrow{Stealth[length=3mm]}}},
      postaction={decorate}] (D) -- (C);
\foreach \x in {\startx,...,\sub} {\node[path_node] at (\x,\starty) {};}
\foreach \y in {\starty,...,\side} {\node[path_node] at (\sub,\y) {};}
\foreach \x in {\sub,...,\side} {\node[path_node] at (\x,\side) {};}
\node[start_node, label={[label_style, below left] $A$}] at (A) {};
\node[end_node, label={[label_style, above right] $C$}] at (C) {};
\end{tikzpicture}
}
\caption{Clean sweep path of subgroup $s=3$ in a $6\times 6$ square with $\gamma=6$ and direction $(+1,+1)$. The path has $10=2(\gamma-1)$ steps and traverses every node of the covered column $x=3$.}

\label{fig:sweep}
\end{figure}
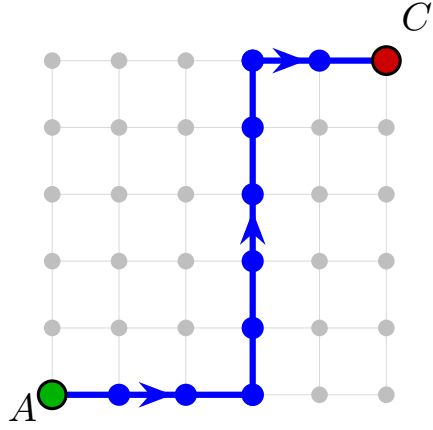

The group is split into $\gamma$ \emph{subgroups}.  Rank the members by identifier so that the largest identifier gets rank~$1$, the next largest rank~$2$, and so on; a smaller rank thus means a higher priority. An agent of rank $\alpha$ goes into subgroup $s=\alpha\bmod\gamma$.  Subgroup $s\in\{0,\dots,\gamma-1\}$ walks an \sharpS-shaped path with three legs (also see Fig. \ref{fig:sweep-head}):

\begin{enumerate}
  \item \emph{first row leg}: $s$ steps along $d_x$, reaching $(a_x+s d_x,\,a_y)$;
  \item \emph{column leg}: $\gamma-1$ steps along $d_y$, reaching $(a_x+s d_x,\,a_y+(\gamma-1)d_y)$;
  \item \emph{second row leg}: $\gamma-1-s$ steps along $d_x$, reaching $C$.
\end{enumerate}

Subgroup $s$ \emph{covers} the column $x=a_x+s d_x$, that is, the $\gamma$ nodes $(a_x+s d_x,\,a_y+j d_y)$ for $j=0,\dots,\gamma-1$; an agent settles only on a node its own subgroup covers (Section~\ref{subsec:spread}).  If a path hits a boundary, the subgroup keeps going toward $C$ along the boundary; each path loses exactly $t_x+t_y$ steps this way but still ends at $C$ (Lemma~\ref{lem:sync}).  Fig.~\ref{fig:sweep} shows one such path.

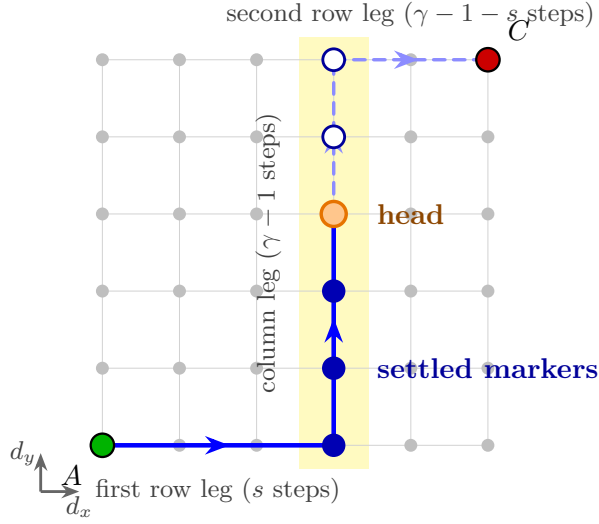
\begin{figure}[t]
\centering
\resizebox{0.5\columnwidth}{!}{%
\begin{tikzpicture}[
  scale=0.95,
  gnode/.style={circle, fill=gray!50, inner sep=1.6pt},
  settled/.style={circle, fill=blue!70!black, inner sep=2.9pt},
  headn/.style={circle, draw=orange!90!black, line width=1.2pt, fill=orange!45, inner sep=3.1pt},
  freecov/.style={circle, draw=blue!60!black, line width=1pt, fill=white, inner sep=2.7pt},
  startn/.style={circle, fill=green!70!black, draw, thick, inner sep=2.9pt},
  endn/.style={circle, fill=red!80!black, draw, thick, inner sep=2.9pt},
  walk/.style={blue, ultra thick, line cap=round, line join=round,
      decoration={markings, mark=at position 0.55 with {\arrow{Stealth[length=3mm]}}}, postaction={decorate}},
  togo/.style={blue!45, very thick, line cap=round, dash pattern=on 4pt off 3pt,
      decoration={markings, mark=at position 0.55 with {\arrow{Stealth[length=3mm]}}}, postaction={decorate}},
  dir/.style={black!60, line width=1pt, -{Stealth[length=2.2mm]}},
  lbl/.style={font=\small\bfseries},
  leg/.style={font=\footnotesize, text=black!75}
]
\fill[yellow!30] (2.55,-0.3) rectangle (3.45,5.3);
\draw[step=1cm, color=gray!35, very thin] (0,0) grid (5,5);
\foreach \x in {0,...,5}{\foreach \y in {0,...,5}{\node[gnode] at (\x,\y){};}}
\draw[walk] (0,0) -- (3,0);
\draw[walk] (3,0) -- (3,3);
\draw[togo] (3,3) -- (3,5);
\draw[togo] (3,5) -- (5,5);
\node[settled] at (3,0){};
\node[settled] at (3,1){};
\node[settled] at (3,2){};
\node[headn] at (3,3){};
\node[freecov] at (3,4){};
\node[freecov] at (3,5){};
\node[startn, label={[lbl]below left:$A$}] at (0,0){};
\node[endn, label={[lbl]above right:$C$}] at (5,5){};
\draw[dir] (-0.8,-0.6) -- (-0.3,-0.6) node[leg,below,inner sep=1pt]{$d_x$};
\draw[dir] (-0.8,-0.6) -- (-0.8,-0.1) node[leg,left,inner sep=1pt]{$d_y$};
\node[leg, anchor=north] at (1.5,-0.28){first row leg ($s$ steps)};
\node[leg, rotate=90, anchor=south] at (2.46,2.5){column leg ($\gamma-1$ steps)};
\node[leg, anchor=south] at (4,5.28){second row leg ($\gamma-1-s$ steps)};
\node[lbl, text=orange!55!black, anchor=west] at (3.42,3.0){head};
\node[lbl, text=blue!60!black, anchor=west] at (3.42,1.0){settled markers};
\end{tikzpicture}}
\caption{A mid-sweep snapshot of one subgroup's path ($s=3$, $\gamma=6$, direction $(+1,+1)$) from the anchor $A$ to the far corner $C$, with the three legs labelled in the figure. The \emph{head} (orange) holds the still-moving agents, with settled markers behind it (filled) and unreached covered nodes ahead (hollow).}
\label{fig:sweep-head}
\end{figure}

As a subgroup walks its path, it settles agents on free nodes and leaves them behind, while its still-moving agents keep advancing together (Fig.~\ref{fig:sweep-head}). At any round, these still-moving agents all sit on one node, namely the farthest node of the path the subgroup has reached so far; we call this node the \emph{head} of the subgroup. The settled markers therefore lie on the stretch of path already walked, and the head is the front end of the subgroup, just beyond them. When a settled marker later refers to ``the step at which the head left it,'' it means the step at which the still-moving agents moved off that settled marker's node.

The next lemma shows that the $\gamma$ covered columns together cover the whole square with no overlap.

\begin{lemma}
\label{lem:coverage}
For a clean sweep with $\gamma\ge 2$, the $\gamma$ covered columns partition the square
$[a_x,\,a_x+(\gamma-1)d_x]\times[a_y,\,a_y+(\gamma-1)d_y]$ into $\gamma$ columns of $\gamma$ nodes each. Consequently, each of the $\gamma^2$ nodes is covered by exactly one subgroup, and every node of the square is traversed by at least one subgroup path.
\end{lemma}

\begin{proof}
Subgroup $s$ covers column $x=a_x+s d_x$, whose cells are $(a_x+s d_x,\,a_y+j d_y)$, $j=0,\dots,\gamma-1$.  As $s$ ranges over $\{0,\dots,\gamma-1\}$, these columns are disjoint and cover the square, so each cell is covered exactly once.  The first row leg of subgroup $\gamma-1$ traverses the anchor row, the second row leg of subgroup $0$ traverses the top row, and the column legs traverse every interior column; hence every cell is traversed by some path. 
\end{proof}

\paragraph{\textbf{Agent state.}}\label{subsec:state}
Each agent stores the variables listed in Table~\ref{tab:vars}, and Theorem~\ref{thm:main} shows that they fit in $O(\log k)$ bits. We refer to the part that fixes the geometry of the agent's current sweep as its \emph{sweep variables}: the sweep parameter $\gamma$, the direction $\delta$, and the anchor $A$ and destination $C$, which are kept as displacements from the agent's current node rather than as absolute grid coordinates. Together with its label, subgroup index $s$, step count, epoch, and, for a settled marker, the settling round $r_{\mathrm{set}}$, these let an agent reconstruct its sweep locally.

\begin{table}[t]
\centering
\small
\renewcommand{\arraystretch}{1.1}

\begin{tabular}{|l|p{0.7\columnwidth}|}
\hline
\textbf{Variable} & \textbf{Description}\\
\hline
\texttt{a.id}        & unique integer identifier of the agent.\\\hline
\texttt{a.status}    & one of \texttt{moving}, \texttt{claiming}, \texttt{settled}, \texttt{terminated}.\\\hline
\texttt{a.label}     & current label $\Lab=(\ell,\cert,\id)$ of the agent's group.\\\hline
\texttt{a.mass}      & active size $m$ used to choose the current cycle's $\gamma$.\\\hline
\texttt{a.subgroup}  & index $s\in\{0,\dots,\gamma-1\}$ of the sweep sub-path the agent walks.\\\hline
\texttt{a.step}      & number of steps taken so far on the current sweep path.\\\hline
\texttt{a.settleRound} & for a settled marker, the round in which it was placed or last claimed, together with the step at which the sweep head left it; these values determine marker age locally.\\\hline
\texttt{a.dir}       & current sweep direction $\delta=(d_x,d_y)\in\{-1,+1\}^2$.\\\hline
\texttt{a.anchor}    & anchor $A$ and (clipped) destination $C$ of the current sweep, stored relative to the agent's position.\\\hline
\texttt{a.contact}   & one bit for the current subgroup path: \texttt{clean} until that path actually meets a strictly higher label, and \texttt{contacted} afterwards.\\\hline
\texttt{a.epoch}     & bounded epoch counter used to match a settled marker, late adopters, and the corresponding wait/recovery window.\\\hline
\end{tabular}%

\caption{Persistent state of an agent. A passive-settled agent stores its sweep variables, subgroup, step, and settling round, enabling the marker rule and later recruitment by a same- or higher-labelled \textsc{Collect}.}
\label{tab:vars}
\end{table}

 \paragraph{\textbf{Initialisation (round~0).}}
At round~$0$, look at each set of agents that start on the same node.  They pick the one with the largest identifier as leader and become a single group with $\ell=0$, $\cert=m$ (the number of agents on that node), and $\id$ equal to the leader's identifier.  Every group starts with direction $(+1,+1)$, epoch~$0$, and a clean contact flag.

\paragraph{\textbf{Singleton rule.}}
If an agent is about to start or continue a phase, is the only moving agent on its node, and has $m=1$, it immediately becomes passive-settled, storing its current label, sweep variables, and settling round.  A single agent on its own is already dispersed, unless a later \textsc{Collect} of the same or a higher label picks it up.

\noindent
Finally, we fix the unit we use to measure progress.  A \emph{cycle} of a group is one \textsc{Spread}, plus a \textsc{Reverse} and a \textsc{Collect} if some agents are still unsettled after the \textsc{Spread}; a \textsc{Spread} that settles everyone is a cycle by itself.  The three phases come one after another in Sections~\ref{subsec:spread}--\ref{subsec:collect}.  
An \emph{epoch} of a label is the longer span during which a group keeps one label value $(\ell,\cert,\id)$: it starts when the group first takes that value and ends when the group levels up, merges into a higher label, or finishes. One epoch may therefore contain several sweep attempts, for instance across bounces. The counter \texttt{a.epoch} records the current epoch, so a settled marker and a later visitor can check whether they refer to the same attempt. Both the cycle and the epoch differ from the \emph{CCM cycle} of Section~\ref{sec:model}.
Two facts about labels will be used later.  First, when a group levels up, it drops its old certificate and sets the new one to the active size counted at $C$.  Second, we track per subgroup path whether it has met a higher label: once a path meets one it is marked \texttt{contacted}, and a contacted path follows that higher label and can no longer raise its own old label, while a path that never met a higher label can still use its old label in the end-of-cycle decision.

\subsection{The \textsc{Spread} Phase}\label{subsec:spread}
\textsc{Spread} is the phase in which an agent first becomes settled, so we fix the term here. A \emph{settled marker} is an agent that has taken a node as its final position and stopped; it keeps the label and sweep record it held when it settled, so any agent that later reaches its node can read them. We use the single term \emph{settled marker} throughout for such an agent, and what a lower-labelled agent does when it reads a higher-labelled settled marker (\emph{late adoption}) is covered in Section~\ref{subsec:interactions}.

At the start of \textsc{Spread}, all $m$ moving members of the current active group are co-located at $A$.  The rank rule fixes the $\gamma$ subgroups and the sweep $A\to C$ is computed.  During the next $\sweeplen$ rounds, every nonempty subgroup advances one step per round.  When a member of subgroup $s$ enters a node \emph{covered by $s$}, a one-round local claim handshake decides settlement:
\begin{itemize}
  \item \textbf{Free node.}  The lowest-identifier still-moving member of subgroup $s$ becomes a \texttt{claiming} candidate for one CCM round, carrying the group's label, sweep variables, subgroup index, current step, and settling round.  In the next communication step, all co-located candidates compare labels.  The maximum label wins; if it is in \textsc{Spread}, exactly the lowest-identifier winner-labelled candidate becomes passive-settled and all other candidates keep moving.  If the maximum label is already in \textsc{Reverse} or \textsc{Collect}, no one settles and the candidates follow that higher phase.
  \item \textbf{Settled lower-labelled marker.}  The visiting group's label and sweep variables overwrite those of the settled agent, which stays passive-settled and is now \emph{claimed}.  No visitor settles here; the subgroup walks on.
  \item \textbf{Settled equal-labelled marker.}  The node is blocked; the subgroup walks on without settling.
  \item \textbf{Settled higher-labelled marker.}  After the CCM exchange, the visitors and the settled marker carry the higher label.  If the stored higher-label epoch is still in \textsc{Spread}, the lowest-identifier co-located higher-labelled agent remains the settled representative and the others follow the settled marker rule of Section~\ref{subsec:marker}.  If the stored epoch is already in \textsc{Reverse} or \textsc{Collect}, the visitors do not settle; they follow that higher phase.  Only the visiting subgroup path records \texttt{contacted}; non-co-located subgroups of the old label are unchanged.
\end{itemize}
A member never settles on a row-leg node covered by another subgroup; coverage and the claim handshake make the settling agent unambiguous.

One might ask when a subgroup even meets an \emph{equal}-labelled settled marker during \textsc{Spread}, since two groups that carry equal labels are in fact one and the same group. This happens \emph{within} a single group: on its row legs a subgroup passes over nodes of other columns, where a sibling subgroup of the same label may already have settled an agent; and it also happens across bounces and later cycles, where a marker the same label left on an earlier attempt is met again. In each case the node is simply treated as blocked and the subgroup walks on.

\paragraph{\textbf{Two moving groups meet.}} If two moving groups of distinct labels become co-located during \textsc{Spread}, they merge locally.  Let $\Lab_{\mathrm{hi}}\succ\Lab_{\mathrm{lo}}$ be the two labels.  Every co-located settled loser is claimed, meaning its label is overwritten by $\Lab_{\mathrm{hi}}$; every co-located moving loser adopts $\Lab_{\mathrm{hi}}$.  A moving loser sharing a node with a winner subgroup joins that subgroup at the winner's current step; if it meets only a winner settled marker, it uses Section~\ref{subsec:marker}.  If the winner is in \textsc{Reverse} or \textsc{Collect}, absorbed agents follow that phase immediately.  Non-co-located agents of the losing label are unchanged: they settle, later meet the higher label themselves, or form an overflow remnant.  The winner's sweep object --- its $\gamma$ and subgroup partition --- is \emph{not} rebuilt mid-cycle; all size and certificate updates wait until the winner reaches $C$.

When the $\sweeplen$ rounds end, let $t_x,t_y$ be the boundary excesses.
\begin{itemize}
  \item If $t_x>0$ or $t_y>0$ (a clipped sweep), the group skips \textsc{Reverse} and \textsc{Collect}; see Section~\ref{subsec:boundary}.
  \item If $t_x=t_y=0$ and no moving member is at $C$, every member is passive-settled and the group has no active movers.
  \item Otherwise the moving members at $C$ form the \emph{overflow}, and the group proceeds to \textsc{Reverse}.
\end{itemize}

\subsection{The \textsc{Reverse} Phase}\label{subsec:reverse}
The \emph{number} $\gamma$ of subgroups stays the same; the overflow members are re-assigned to subgroup indices by their new rank (the same $s=\alpha\bmod\gamma$ rule), and each nonempty subgroup retraces its path to $A$ in $\sweeplen$ rounds.
The overflow can be small --- as few as one agent --- so when it has fewer than $\gamma$ members the ranks fall into that many distinct indices and the remaining indices are simply \emph{empty}; only nonempty subgroups move. This is the same empty-subgroup situation as in \textsc{Spread} when $m<\gamma$, so it needs no special handling.
\textsc{Reverse} is a moving-agent rendezvous phase, not a final-settlement phase.  On each covered node, the subgroup absorbs:
\begin{itemize}
  \item any \emph{settled marker} of lower label, whose label and sweep variables are overwritten but which remains settled until a later \textsc{Collect}, and it does not start moving now;
  
  \item any \emph{moving} agent of equal or lower label, which leaves its group, joins and  travels with the subgroup back toward $A$.
\end{itemize}
If the group meets a strictly \emph{higher}-labelled moving agent, it loses the merger and follows the winner as in Section~\ref{subsec:spread}; if it meets only a higher-labelled settled marker, it follows Section~\ref{subsec:marker}.  At the end of \textsc{Reverse}, every surviving moving member is at $A$.  Thus \textsc{Reverse} synchronizes the overflow and any moving agents it meets; settled markers remain final unless the following \textsc{Collect} reaches them.

\subsection{The \textsc{Collect} Phase}\label{subsec:collect}
Again the \emph{number} $\gamma$ of subgroups is unchanged; the members at $A$ are re-ranked and re-assigned to subgroup indices by their new rank, and each subgroup walks $A\to C$ once more in $\sweeplen$ rounds.  On covered nodes, \textsc{Collect} absorbs every settled or moving agent of equal or lower label that it reaches; settled agents un-settle and join the collecting subgroup.  A higher-labelled moving encounter makes the visitor lose a merger, and a higher-labelled settled marker invokes Section~\ref{subsec:marker}.  \textsc{Collect} settles no one; it is, however, the only phase that can recruit final settled markers into the same- or higher-labelled moving group.  When it ends, every surviving collected member is co-located at $C$.

\subsection{Higher-Label Interactions}\label{subsec:interactions}
Every interaction in the three phases (Sections~\ref{subsec:spread}--\ref{subsec:collect}) follows one rule: \emph{the higher label wins}.  Table~\ref{tab:mergers} shows, for each phase, what happens when a higher-labelled agent meets a lower-labelled one.  The rest of this subsection explains two cases in more detail: what happens when a higher-labelled group overlaps only part of a lower group (\emph{partial absorption}), and what a lower-labelled agent does when it first meets a higher label through a settled marker (\emph{late adoption}).

\begin{table}[t]
\centering
\small
\renewcommand{\arraystretch}{1}

\begin{tabular}{|p{0.18\columnwidth}|p{0.24\columnwidth}|p{0.43\columnwidth}|}
\hline
\textbf{Higher-label object encountered} & \textbf{Lower-label \emph{settled} agent at the node} & \textbf{Lower-label \emph{moving} agent at the node}\\
\hline
Moving \textsc{Spread} & Label and sweep variables overwritten in place; the agent stays passive-settled unless a later same- or higher-labelled \textsc{Collect} recruits it. & Co-located loser agents join the winner's subgroup at the winner's step; non-co-located loser agents are unchanged until they meet the winner or finish their own sweep.\\
\hline
Moving \textsc{Reverse} & Lower settled marker is claimed but stays passive-settled. & Co-located lower moving agent adopts the higher label and follows the winner's reverse phase.\\
\hline
Moving \textsc{Collect} & Lower settled agent un-settles and joins the collector. & Co-located lower moving agent adopts the higher label and follows the winner's collect phase.\\
\hline
Settled marker & No interaction between two settled agents. & The lower moving agent adopts the higher label, uses the settled marker's settling round to predict the stored subgroup head, and either joins the still-spreading head, follows the predicted reverse/collect phase, waits at $C$ for an ordinary collect, or starts the unique recovery \textsc{Reverse--Collect} for that label/epoch.\\
\hline
\end{tabular}%

\caption{Resolution of a label clash between two co-located agents. A settled agent stores the sweep variables, subgroup, step, epoch, and settling round of the latest active label, enabling local late adoption from a higher settled marker.}
\label{tab:mergers}
\end{table}

\paragraph{\textbf{Partial absorption and size decrease.}}
A group occupies a $\gamma\times\gamma$ square as $\gamma$ subgroups, so a strictly higher-labelled group $H$ may cover only \emph{part} of it.  The subgroups of $G$ inside the overlapping part meet $H$ directly, step onto nodes $H$ has settled or claimed, or become late adopters (described below); these agents move up to $H$ and their subgroup paths become contacted.  The subgroups outside the overlapping part either settle on free nodes or, if enough of their nodes are blocked to leave agents still moving at $C$, form a smaller overflow of the old label.  So one sweep can split $G$ into three parts: some agents join $H$, some are already settled, and at most the overflow keeps going.  The contacted part keeps $G$'s old certificate and cannot raise the old label.  A part that never met the higher label may still promote at $C$ if it gathers enough agents there.  So any drop in a group's active size is either agents moving up to a higher label or agents settling for good; it is never a reason to level up.

\paragraph{\textbf{Late adoption from a higher settled marker.}}\label{subsec:marker}
A settled marker is not a shared memory cell, but it does store a record of the sweep that placed or last claimed it: the label $\Lab$, the cycle parameter $\gamma$, the anchor $A$, the ideal corner $C^\star$, the direction $\delta$, the subgroup index $s$, the step $q$ at which the subgroup head left it, the epoch, and the settling round $r_{\mathrm{set}}$. It stores the ideal corner $C^\star = (a_x+(\gamma-1)d_x, a_y+ (\gamma-1)d_y)$, not the realised destination $C$: the ideal corner is fixed by $A$, $\gamma$ and $\delta$ at the very start of the sweep, so every member knows it even if the sweep is later clipped. whereas an agent that settled before the boundary was discovered need not know the realised $C$. When a lower-labelled group of moving agents enters the settled marker's node at round $r$, it computes $r-r_{\mathrm{set}}$ and works out which phase the higher-labelled subgroup head is in and where it is on the stored path.

These agents at once adopt the higher label, set their old subgroup-contact flag to \texttt{contacted}, and leave their old group.  What they do next depends on the phase they predicted:
\begin{itemize}
  \item If the \emph{head} is predicted to be in \textsc{Spread}, they walk along the stored subgroup path toward the predicted head.  If they reach the head before $C$, they join that subgroup.  If the subgroup still has enough empty covered nodes left, no overflow is created; if some of the incoming agents are still moving after the last covered node, they find the overflow themselves at $C$.
  \item If the \emph{head} is predicted to be in \textsc{Reverse} or \textsc{Collect}, they join that same higher-labelled pass on the stored path, rather than starting anything for the old label.  Concretely, each such agent begins the predicted \textsc{Reverse} (or \textsc{Collect}) from its current node and moves along it in that phase.  Several late adopters may walk the same labelled pass on the same path at once; this is harmless, because they all carry the same label and meet at the common endpoint of that pass.
  \item This third case is not a separate rule but what happens when the previous cases fail to connect: if the adopters miss the live reverse/collect traffic --- for example, by crossing it on an edge, walking one sweep path while the reverse runs on another, or arriving after the head has already left that stretch --- they keep going to $C$ with the higher label and wait there.
\end{itemize}

At $C$, all agents that adopted the same label and epoch wait up to $2\gamma$ rounds for a \textsc{Collect} caused by overflow of that higher-labelled sweep.  If such a \textsc{Collect} reaches $C$, they join it, and no extra cycle is created.  Otherwise the sweep has no overflow still moving.  The waiting agents pick the largest ID among themselves and start exactly one recovery \textsc{Reverse--Collect} using the stored sweep variables.
By a \emph{recovery \textsc{Reverse--Collect}} we mean an ordinary \textsc{Reverse} followed by an ordinary \textsc{Collect} (as in Sections~\ref{subsec:reverse} and~\ref{subsec:collect}), with two differences: it is started by the waiting late-adopters rather than by a group's own overflow, and it uses the marker's stored sweep variables for its subgroups and path.
This recovery does not promote: it may collect same- or lower-labelled agents, but it can raise the level or certificate only by merging into a live higher-labelled group whose own end-of-cycle decision allows the increase.

If the higher-labelled sweep had no overflow, the same prediction instead leads these agents to the settled head of the subgroup path.  So one local rule covers both cases: either they find the head --- still spreading or already settled --- and join or settle without a new cycle, or they reach $C$, where the $2\gamma$ wait tells an ordinary overflow collection apart from the one recovery cycle.  This is why correctness does not need edge-crossing detection or every subgroup path to be occupied at the same time.

A prediction can also be wrong because the higher sweep was \emph{clipped}: the adopters predict on the ideal path and so may expect a longer sweep than actually happened. This needs no special case. Walking the ideal path, they meet the same boundary the original sweep met, stop at the realized corner, and the $2\gamma$ wait applies there exactly as above --- they join an ordinary \textsc{Collect} if one arrives, and otherwise trigger the single recovery cycle. So every misprediction, clipped or not, falls through to the wait at the realized corner and at most one recovery cycle.

\subsection{Boundary Handling}\label{subsec:boundary}
If \textsc{Spread} is clipped, the intended square reaches the grid boundary, as illustrated in Fig.~\ref{fig:bounce}. Since this does not reflect any change in the group's size, \textsc{Reverse} and \textsc{Collect} are skipped. The agent then distinguishes two cases locally, depending on whether the same $\gamma$-square fits after re-anchoring at the hit corner.

\paragraph{\textbf{Bounce (the sweep fits in the interior, $\gamma\le\Side$).}}
The active subgroup performs a \emph{boundary bounce}, continuing the same \textsc{Spread} attempt:
\begin{itemize}

  \item set $A\leftarrow C$ (the clipped corner);
  \item flip each boundary-hit axis: $d_x\leftarrow-d_x$ if $t_x>0$, $d_y\leftarrow-d_y$ if $t_y>0$, leaving the other axis unchanged;
  \item keep $\ell,\cert,\id,m,\gamma$, the epoch, and the subgroup-contact flags unchanged;
  \item recompute the sweep from $A$ and begin the next \textsc{Spread}.
\end{itemize}
Settled markers placed before boundary discovery remain passive-settled; only the still-moving active subgroup continues from the clipped corner.  Because $\gamma$ is not recomputed, the intended square is simply reflected across the boundary-hit axis, and the certificate is unchanged.  After a bounce on an axis, the next \textsc{Spread} moves \emph{away} from that boundary, so the same boundary cannot be hit again immediately.  Boundary contact is therefore only a \textsc{Spread}-phase event, and \textsc{Reverse}/\textsc{Collect} never clip.

\paragraph{\textbf{Terminating whole-grid sweep ($\gamma>\Side$).}}
If a \textsc{Spread} starting from a grid corner still clips on \emph{both} axes, then $\gamma$ exceeds the grid side length. Agents detect this locally because both required outgoing ports are absent, so flipping the sweep direction cannot make the square fit inside the grid. Since the intended square is at least as large as the grid, the group settles all its $m\le k\le n$ active members over the $n$ nodes according to the coverage rule, and all agents terminate. By Lemma~\ref{lem:threecycles} and Theorem~\ref{thm:main}, this case occurs only at the $O(1)$ largest levels and takes $O(\Side)=O(\sqrt{k})$ rounds.

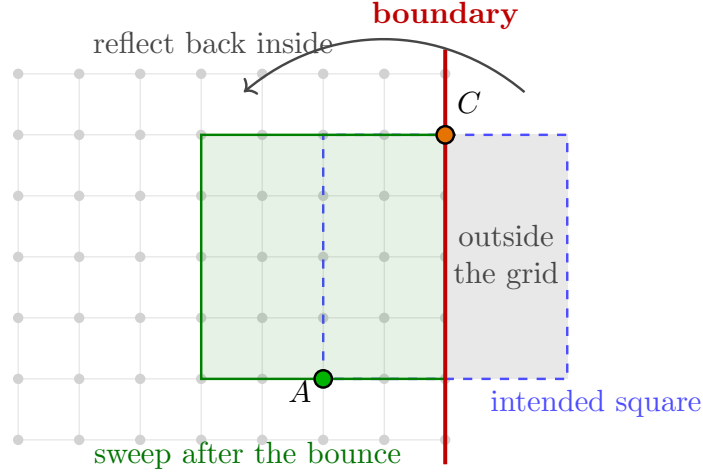
\begin{figure}[t]
\centering
\resizebox{0.6\columnwidth}{!}{
\begin{tikzpicture}[
    scale=0.7,
    gnode/.style={circle, fill=gray!35, inner sep=1.2pt},
    wall/.style={red!75!black, line width=1.3pt},
    intended/.style={blue!70, thick, dashed},
    reflected/.style={green!50!black, thick},
    ptA/.style={circle, fill=green!70!black, draw, thick, inner sep=2pt},
    ptC/.style={circle, fill=orange!90!black, draw, thick, inner sep=2pt},
    lab/.style={font=\small}
]
\def\W{7}\def\H{6}
\draw[step=1cm, gray!25, very thin] (0,0) grid (\W,\H);
\foreach \x in {0,...,\W}{\foreach \y in {0,...,\H}{\node[gnode] at (\x,\y){};}}

\fill[gray!18] (7,1) rectangle (9,5);
\draw[intended] (5,1) rectangle (9,5);

\fill[green!50!black, fill opacity=0.10] (3,1) rectangle (7,5);
\draw[reflected] (3,1) rectangle (7,5);

\draw[wall] (\W,-0.4) -- (\W,\H+0.4);

\node[ptA, label={[lab,below left]$A$}] at (5,1){};
\node[ptC, label={[lab,above right]$C$}] at (7,5){};

\draw[->, thick, gray!55!black] (8.3,5.7) to[bend right=40] (3.7,5.7);

\node[red!75!black, anchor=south, font=\small\bfseries] at (\W,\H+0.55) {boundary};
\node[gray!55!black, lab, anchor=south] at (3.2,6.15) {reflect back inside};
\node[gray!55!black, lab, align=center] at (8,3) {outside\\the grid};
\node[blue!70, lab, anchor=west] at (7.55,0.6) {intended square};
\node[green!45!black, lab, anchor=west] at (1.05,-0.25) {sweep after the bounce};
\end{tikzpicture}
}
\caption{A \emph{boundary bounce}. The group's intended sweep (blue dashed) starts at $A$ and would run past the right boundary, so part of it (grey) lies outside the grid. On reaching the boundary at $C$, the group reflects the square back inside (green) and sweeps that instead, without ever needing to know the grid size.}
\label{fig:bounce}
\end{figure}

\subsection{Decision at the End of a Clean Cycle}\label{subsec:decision}
After a clean \textsc{Collect}, let $m_{\mathrm{after}}$ be the active size recomputed from the agents at $C$, and set $\gamma_{\mathrm{new}}=2\ceil*{\sqrt{m_{\mathrm{after}}}}$.  The group updates:
\begin{itemize}
  \item $m\leftarrow m_{\mathrm{after}}$,\ \ $\id\leftarrow$ largest identifier among the members,\ \ $\gamma\leftarrow\gamma_{\mathrm{new}}$
  \item $A\leftarrow C$
  \item flip \emph{both} axes: $(d_x,d_y)\leftarrow(-d_x,-d_y)$
  \item if the subgroup path is clean and $m_{\mathrm{after}}\ge 2\cert$, set $\cert\leftarrow m_{\mathrm{after}}$ and increment $\ell\leftarrow\ell+1$ (a \emph{level-up})
  \item if the subgroup path has contacted a strictly higher label, keep $\ell$ and $\cert$ unchanged and treat this as a losing lower-label path; it then run only the bounded recovery steps of the late-adoption rule (Section \ref{subsec:marker}), or a final settling \textsc{Spread}, after which its remaining agents have settled or moved up.
  \item begin the next \textsc{Spread}.
\end{itemize}
Thus a higher-label merger is certified only at $C$: the winning label increases its level only after \textsc{Collect} recomputes at least twice the current certificate there.  Conversely, a lower-label subgroup that lost agents to a higher label cannot use its smaller active population as a new promotion base; it remains a losing label with unchanged certificate.  A clean subgroup that did not meet any higher subgroup may still promote the old label if it reaches $C$ with enough collected agents.  Since promotion is decided only at the end of a cycle, any contact that can block it must have occurred during that cycle's first \textsc{Spread}.

The unconditional double flip is the geometric keystone.  In a clean cycle, the nodes relevant to the continuing active group have either been collected or are already occupied by final same-label settled agents.  Flipping both axes sends the next sweep back through the region whose slack was just exposed, and after a level-up extends beyond it.  In a contacted losing path, the same flip is used only for bounded recovery or final settling; the remnant cannot promote itself after transferring agents to a higher label.

\begin{algorithm2e}
\small
\caption{One round of agent $a$ under SRC.}
\label{alg:src}
\SetKwInOut{Input}{Input}
\Input{state of agent $a$; set $S_a$ of agents co-located with $a$.}
$\Lab^\star\leftarrow$ the maximum label among $S_a\cup\{a\}$\tcp*{higher-label wins}
\If{$a.\textit{label}\prec\Lab^\star$}{
  set this subgroup path's contact flag to \texttt{contacted}\;
  adopt $\Lab^\star$ and the sweep variables of a co-located holder of $\Lab^\star$\;
  \eIf{the holder is moving}{
    join the holder's subgroup at its current step and phase\;
  }{
    follow the settled marker rule of Section~\ref{subsec:marker}\;
  }
}
\If{$a.\textit{status}=\texttt{claiming}$}{
  compare labels and identifiers with all co-located claiming agents\;
  \eIf{$a$ is the lowest-identifier claimant of the maximum \textsc{Spread} label}{
    $a.\textit{status}\leftarrow\texttt{settled}$; store label, subgroup, step, epoch, and settling round\;
  }{
    $a.\textit{status}\leftarrow\texttt{moving}$; follow the maximum label's phase\;
  }
}
\If{$a.\textit{status}=\texttt{settled}$}{
  retain the stored settling round and expose the derived settled marker age to visitors\;
  \If{an equal- or higher-label \textsc{Collect} sweep enters $a$'s node}{
    $a.\textit{status}\leftarrow\texttt{moving}$;\ join the visiting subgroup\;
  }
  \Return\;
}
\Switch{$a$'s current phase}{
  \Case{\textsc{Spread}}{
    advance one step on subgroup $s$'s path\;
    on entering a covered node: apply the claim/settling rule (Section~\ref{subsec:spread})\;
    if $\texttt{step}=\sweeplen$: handle end of \textsc{Spread} --- overflow, clip, or no active movers\;
  }
  \Case{\textsc{Reverse}}{
    advance one step backward on subgroup $s$'s path\;
    on entering a covered node: absorb moving equal- or lower-label agents and claim lower settled markers (Section~\ref{subsec:reverse})\;
    if $\texttt{step}=\sweeplen$: switch to \textsc{Collect}; re-partition the members at $A$\;
  }
  \Case{\textsc{Collect}}{
    advance one step on subgroup $s$'s path\;
    on entering a covered node: absorb equal- or lower-label agents\;
    if $\texttt{step}=\sweeplen$: perform the end-of-cycle decision (Section~\ref{subsec:decision})\;
  }
}
\end{algorithm2e}

\section{Analysis of the Algorithm}\label{sec:analysis}

We analyze SRC by establishing synchronization of its phases, proving monotonicity and the doubling invariant for levels, and showing that each subgroup path participates in only $O(1)$ SRC cycles before it settles, levels up, or transfers to a higher label. This yields the desired time and memory bounds.

\begin{restatable}{lemma}{SyncLemma}\label{lem:sync}
Every \textsc{Spread}, \textsc{Reverse}, or \textsc{Collect} phase lasts exactly $\sweeplen=2(\gamma-1)-t_x-t_y\le 2(\gamma-1)=O(\sqrt{m})$ rounds, and all still-moving members assigned to the phase arrive simultaneously at $C$ (at $A$ for \textsc{Reverse}). The late-adopter wait/recovery window lasts $2\gamma=O(\sqrt{m})$ rounds.
\end{restatable}

\begin{proof}
By Section~\ref{subsec:sweep}, an unclipped subgroup path has length $s+(\gamma-1)+(\gamma-1-s)=2(\gamma-1)$, independent of $s$.  Under clipping, a subgroup that reaches the boundary in $x$ continues toward $C$ along the boundary, replacing the missing $t_x$ row-leg steps; the $y$ case is symmetric.  Checking the three legs shows that each path loses exactly $t_x$ row-leg steps and $t_y$ column-leg steps, so each is shortened by $t_x+t_y$ and still ends at $C$.  Thus, all paths have length $\sweeplen=2(\gamma-1)-t_x-t_y$, and members still moving after their last settling/absorption arrive together.  Since $\gamma=2\ceil*{\sqrt{m}}\le 2(\sqrt{m}+1)$, we have $\sweeplen\le 2(\gamma-1)=O(\sqrt{m})$.  A late-adopting agents use one suffix of such a path plus a $2\gamma$ return window, giving the same asymptotic order. 
\end{proof}

Next, we show that levels are monotone and satisfy the doubling invariant.

\begin{restatable}{lemma}{MonotoneLemma}\label{lem:mono}
The level $\ell$ and certificate $\cert$ carried by an agent never decrease. A label's active size may decrease when agents settle or transfer to a strictly higher label, but that decrease never justifies a level-up. A subgroup path that has contacted a higher label is losing, keeps the old label, and cannot raise $\ell$ or $\cert$.
\end{restatable}

\begin{proof}
An agent's level changes in two ways: a clean doubling cycle increases $\ell$ by $1$, or a merger gives the agent the winner's label $\Lab_{\mathrm{hi}}\succ\Lab_{\mathrm{lo}}$, whose level is no smaller by lexicographic order.  The certificate $\cert$ changes only at such a clean level-up, after the collected size is recomputed at $C$, or by adoption of a higher label; hence it also never decreases.

For active size, fix one cycle of a label episode $G$.  If a subgroup of $G$ meets no higher label, its end-of-cycle decision uses only the members present at $C$ after \textsc{Collect}.  If that recomputed size is large enough, it certifies a level-up; otherwise the subgroup runs only bounded final/recovery steps.  If a subgroup becomes co-located with a strictly higher-labelled group $H$, enters a higher marker, or is collected by $H$, exactly those encountered agents adopt $H$ and leave $G$.  Agents of $G$ that settle are already dispersed, and only overflow can continue.  Thus every decrease is either final settlement or upward transfer.  In the upward-transfer case, the affected subgroup path is contacted and the decision rule forbids a level-up from the smaller base. 
\end{proof}

\begin{restatable}[doubling invariant]{lemma}{DoublingLemma}\label{lem:doubling}
Whenever a label \emph{attains} level $\ell$ --- at birth for $\ell=0$, or by the level-up that creates it --- its certificate satisfies $\cert\ge 2^\ell$. Hence no agent ever exceeds level $\lceil\log_2 k\rceil$. A contacted losing remnant may later have active size below $2^\ell$, but it keeps the old certificate and cannot use the smaller active size as a promotion base.
\end{restatable}

\begin{proof}
Induction over the events that change a label's level.  Initially $\ell=0$ and $\cert=m\ge 1=2^0$.

\emph{Merger (absorption).}  A loser inherits the winner's label $(\ell_w,\cert_w,\id_w)$ with $\cert_w\ge 2^{\ell_w}$ by hypothesis, so the bound holds for the absorbed agents.

\emph{Size decrease.}  By Lemma~\ref{lem:mono}, a higher-label decrease marks only the affected subgroup path as contacted and leaves the old level and certificate unchanged.  It therefore never creates a freshly levelled group, and the departing agents move to a label whose level already satisfies the bound.

\emph{Level-up.}  Only a clean subgroup path with $m_{\mathrm{after}}\ge 2\cert$ raises the level from $\ell$ to $\ell+1$ and sets $\cert\leftarrow m_{\mathrm{after}}$.  The size is recomputed at $C$ after merge/collect completes.  Since the path is clean, it has met no higher label since attaining level $\ell$; agents counted in $m_{\mathrm{after}}$ are present locally at the promotion decision, not a smaller remnant left after upward transfer.  By induction, the old certificate was at least $2^\ell$, and the clean certificate-doubling test gives $$\cert=m_{\mathrm{after}}\ge2\cert_{\mathrm{old}}\ge 2^{\ell+1}$$.

Finally, at any level-up $\cert=m_{\mathrm{after}}\le k$, so $2^{\ell+1}\le k$ and $\ell+1\le\lceil\log_2 k\rceil$. 
\end{proof}

\begin{restatable}[no lower-label overtake]{lemma}{OvertakeLemma}\label{lem:overtake}
A lower-labelled subgroup cannot gain higher priority from the settled marker of a higher-labelled group. If it encounters that settled marker, it adopts the higher label. Otherwise, its end-of-cycle decision depends only on the agents recomputed at its own $C$.
\end{restatable}

\begin{proof}
The first contact with a higher moving subgroup, higher \textsc{Reverse}/\textsc{Collect}, or higher marker is resolved by Sections~\ref{subsec:spread}--\ref{subsec:marker}: the contacted lower agents adopt the higher label and the old subgroup path becomes losing.  If the higher label is in \textsc{Reverse} or \textsc{Collect}, the adopted agents follow that phase; if the higher object is only a marker, the stored round lets them predict the higher subgroup head and either join it, wait for the ordinary collect at $C$, or start the unique recovery cycle.  In all cases, any level increase is the higher label's increase at its own $C$ after recomputing the collected size.

If no contact occurs, the lower subgroup has not seen the higher marker and cannot count those agents in its local decision.  A much smaller lower subgroup therefore finishes its clean settling/final-spread work on its own path before it can acquire the mass of a larger higher-labelled group; acquiring that mass would require contact, which is the previous case.  Thus a smaller lower label cannot overtake a high-size group at the only place priority can change: the recomputed end-of-cycle decision at $C$. 
\end{proof}

\begin{restatable}[bounded bounces]{lemma}{BounceLemma}\label{lem:bounce}
While $\gamma\le\Side$, a group undergoes at most $4$ boundary bounces before its next clean \textsc{Spread}. If $\gamma>\Side$, it performs at most $4$ bounces and then one terminating whole-grid \textsc{Spread}.
\end{restatable}

\begin{proof}
A bounce on the $x$-axis flips $d_x$ and re-anchors at the hit boundary, so the next \textsc{Spread} moves in direction $-d_x$, away from that boundary; the same $x$-boundary cannot be hit again immediately, and the $y$ case is symmetric.  After at most one bounce per axis, the next \textsc{Spread} heads into the interior on both axes.  If $\gamma\le\Side$, a $\gamma\times\gamma$ square anchored at the resulting corner fits, so the \textsc{Spread} is clean; before this alignment, the worst case is one bounce per axis-direction, at most $4$.  If $\gamma>\Side$, no square fits, the both-axes clip at a corner is detected, and the terminating whole-grid \textsc{Spread} of Section~\ref{subsec:boundary} runs. 
\end{proof}

\begin{restatable}[three cycles per label]{lemma}{ThreeCyclesLemma}\label{lem:threecycles}
Fix a label value $\Lab=(\ell,\cert,\id)$, one epoch of that label, and one subgroup path of that epoch with $\gamma\le\Side$. The path is charged for at most three SRC cycles of label $\Lab$ before it levels up, passive-settles all remaining agents, or transfers all remaining active agents to a strictly higher label. Sweep-path crossings create no extra cycles for the lower label: crossing agents transfer to the higher path and are charged there.
\end{restatable}

\begin{proof}
Let $m_0$ be the active size at the start of the first ordinary \textsc{Spread} of the epoch, let $\cert_0$ be the current certificate, and let $\gamma_0=2\ceil*{\sqrt{m_0}}$, so the square area is $\gamma_0^2\ge4m_0$.  A clean epoch begins before any upward loss for this label, so $m_0\ge\cert_0$; contacted remnants are handled next.

\emph{Clean path.}  Suppose the subgroup path meets no strictly higher label during the first \textsc{Spread}.  Every blocked covered cell is then due to an equal/lower label that is claimed, collected, or ignored as an already-settled same-label marker.  Let $m_1$ be the active size recomputed at $C$ after the cycle.  If $m_1\ge2\cert_0$, the decision rule performs the single permitted level-up at $C$ and the old label episode ends.  This is the promotion case: a subgroup that met no higher subgroup arrives at $C$ with many agents.  Since promotion occurs only at the end of SRC, any higher-label contact capable of preventing it would have happened in the first \textsc{Spread}.  If $m_1<2\cert_0$, then $m_1<2m_0$, so the remaining active population is smaller than half the slack square.  The double flip sends the next \textsc{Spread} back through the same geometric square, where same-labelled settled agents are already safely dispersed and the remaining slack can settle the active remnant unless a new higher-label contact occurs.  Thus a clean path either levels up after the first cycle or settles within the next cycle.

\emph{Contacted path.}  Now suppose the first or second cycle contacts a strictly higher label.  The contact may be a moving--moving collision, a higher \textsc{Reverse}/\textsc{Collect}, or a higher marker.  In all cases, exactly the agents on the contacted path transfer upward and the old path becomes losing.  Old-label agents that already settled remain passive-settled unless a same- or higher-labelled \textsc{Collect} later recruits them.  The only active old-label agents that can remain are overflow agents at the old destination $C$; they keep the old certificate, cannot level up from the reduced size, and run at most one ordinary reverse/collect or final spread to settle the remnant.

It remains to check late adoption from a higher marker.  By Section~\ref{subsec:marker}, the adopting agents know how long ago the stored spread head left that marker.  If it finds the higher subgroup head, then either no overflow was created, or the joined subgroup's excess agents detect overflow at $C$ and wait for the higher label's ordinary collect.  If they cross an already-started reverse phase, follows a path whose reverse runs elsewhere, or misses the reverse/collect altogether, it still reaches $C$ and waits $2\gamma$ rounds.  A collect reaching $C$ absorbs it; if none reaches $C$ in that window, they start the unique recovery \textsc{Reverse--Collect}.  This recovery is charged to the higher label's epoch and is non-promoting unless it merges into a live higher label.

Combining the cases: a clean path uses at most two cycles; a contacted losing path uses the contact cycle, at most one remnant/final-settling cycle, and, for agents that adopted through a marker, at most one recovery cycle charged to the higher label.  Hence, each subgroup path of a fixed label is charged for at most three SRC cycles. 
\end{proof}

\begin{corollary}\label{cor:phases}
For a fixed subgroup path within a single epoch of a label, the path runs $O(1)$ phases: at most $4$ bounce \textsc{Spread}s, at most $3$ cycles (each with one \textsc{Spread}, one \textsc{Reverse}, and one \textsc{Collect}, with terminating cycles possibly shorter), one late-adopter return window, and at most one terminating whole-grid \textsc{Spread}.
\end{corollary}

We now prove the main complexity result.
\begin{theorem}\label{thm:main}
SRC terminates on every oriented square grid.  Its running time is
$T(k) = O(\sqrt{k})$,
and every agent uses $O(\log k)$ bits of memory.
\end{theorem}

\begin{proof}
\textbf{Termination.}  By Lemma~\ref{lem:doubling}, every agent's level is bounded by $\lceil\log_2 k\rceil$; by Corollary~\ref{cor:phases}, each subgroup path of a label episode contributes $O(1)$ phases; and by Lemma~\ref{lem:sync}, each phase and late-adopter wait is finite.  A contacted losing path cannot create a new level, so it cannot repeat indefinitely at the same label.  Hence every agent takes part in finitely many phases.

\textbf{Time bound.}  Fix an agent $a$ and bound its active rounds.  Charge $a$'s work to \emph{dyadic certificate scales}: a phase, wait window, or recovery is \emph{at scale $j$} if the label certificate satisfies $\cert\in[2^{j},2^{j+1})$ during it.  Active size is never larger than $k$, and in a clean level-$\ell$ episode the certificate is backed by the active size that created the level and by the recomputed size at $C$.  By Lemma~\ref{lem:sync}, a scale-$j$ phase lasts $O(\gamma)=O\!\bigl(2^{j/2}\bigr)$ rounds, up to the constant slack in $\gamma$; a contacted remnant can only be smaller, and time spent following a higher marker is charged to the higher label's scale.  Lemma~\ref{lem:overtake} prevents a smaller lower label from borrowing a higher group's mass without becoming part of the higher label.  It remains to bound the number of phases $a$ performs at each scale.

\emph{(i) The scales are bounded by $\lceil\log_2 k\rceil$.}  No certificate exceeds the number of agents collected to create it, so $\cert\le k$ and $j\le\log_2 k$ throughout.

\emph{(ii) $a$ performs $O(1)$ phases at each scale.}  Consider the label episode containing $a$ while its certificate is at scale $j$.  If its subgroup path meets no higher label, Corollary~\ref{cor:phases} gives $O(1)$ phases before it either levels up --- at which point Lemma~\ref{lem:doubling} implies that its certificate has at least doubled and left scale $j$ upward --- or terminates.  If it meets a strictly higher label, the encountered agents transfer upward immediately or through the late-adopter rule.  Agents that do not transfer either settle or form a contacted losing remnant.  By Lemma~\ref{lem:threecycles}, that remnant has only $O(1)$ further phases and cannot level up from the shrunken base.  When $a$ transfers upward, its remaining work is counted in the higher label's three-cycle episode; if $a$ remains in the losing remnant, it performs only bounded remnant work before settling or later transferring.  Hence $a$ occupies any fixed certificate scale for only $O(1)$ phases.

Combining, $a$'s total active time is
\[
  \sum_{j=0}^{\lceil\log_2 k\rceil} O\!\bigl(2^{j/2}\bigr)
  \ =\ O\!\bigl(2^{(\log_2 k)/2}\bigr)\cdot\frac{1}{1-1/\sqrt2}
  \ =\ O\!\bigl(\sqrt{k}\bigr),
\]
the geometric sum being dominated by its top term $\sqrt k$.  This is the \emph{hierarchical merging-and-doubling} bound with subgroup-local losing paths: partial absorption either settles the losing agents or transfers them to a strictly higher label whose three-cycle budget pays for late adoption/recovery.  Since the bound holds for every agent, $T(k)=O(\sqrt k)$.

\textbf{Memory.}  Each agent stores: its identifier ($O(\log k)$ bits); its label $(\ell,\cert,\id)$, where $\ell\le\lceil\log_2 k\rceil$ needs $O(\log\log k)$ bits and $\cert\le k$ and $\id$ need $O(\log k)$ each; active mass $m\le k$; subgroup index $s\le\gamma=O(\sqrt{k})$ and step count $\le\sweeplen=O(\sqrt{k})$, each $O(\log k)$ bits; the marker settling round and recovery timer, capped at $2\gamma=O(\sqrt{k})$ relative to the stored epoch; the subgroup-contact bit; the bounded epoch counter; the $2$-bit direction; and the anchor and destination stored \emph{relative} to the agent's position, whose per-axis displacement is at most $\gamma-1=O(\sqrt{k})$, i.e.\ $O(\log k)$ bits.  The agent never stores an absolute coordinate and never queries $n$.  The total is $O(\log k)$ bits. 
\end{proof}

Finally, we prove the correctness.
\begin{restatable}{lemma}{NoConflictLemma}\label{lem:noconflict}
At every round at most one settled agent occupies any node.
\end{restatable}

\begin{proof}
A node becomes occupied by a settled agent only after the claim handshake on a covered-cell entry during \textsc{Spread}, or by retaining an already passive-settled marker whose label is overwritten.  In the claim handshake, all co-located candidates compare labels and identifiers before any candidate settles, and exactly the lowest-identifier candidate among the maximum \textsc{Spread} label may settle.  If a settled agent of equal or higher label is present, no second settled agent is created; the marker rule keeps exactly one settled representative at the node.  If a settled agent of strictly lower label is present, its label is overwritten in place.  Thus the number of settled agents on a node is always in $\{0,1\}$. 
\end{proof}

\begin{restatable}[overlap resolution]{lemma}{MeetingLemma}\label{lem:meeting}
If two groups $G_1$ and $G_2$ have overlapping swept squares and both attempt to settle on the same node, the lower label either meets the higher moving group, encounters a higher settled marker, or treats the higher settled agent as a blocker. In every case, at most one agent settles on the node, and every lower-labelled active agent that encounters the higher label transfers upward, marking only its own old subgroup path as contacted.
\end{restatable}

\begin{proof}
Let $u$ be a cell common to the two attempted sweeps.  If the two moving visits to $u$ coincide in a CCM round, the agents communicate and the lower label adopts the higher label by the merger rule; if both labels try to settle there, the claim handshake lets only the lowest-identifier agent among the maximum label settle.  If the visits do not coincide, the later visitor either finds no settled agent, leaving the cell free for the first settling attempt that reaches it, or finds a marker left by the earlier sweep.  A lower marker is overwritten or collected by the higher label according to the phase.  A higher marker blocks lower-label settlement and invokes the late-adopter rule of Section~\ref{subsec:marker} for the lower moving agents.  By Lemma~\ref{lem:noconflict}, no step creates a second settled agent at $u$.  Thus overlap is resolved locally: lower active agents that see the higher label move upward, and lower agents that never see it either settle elsewhere or continue as a clean or contacted overflow remnant according to their own subgroup histories. 
\end{proof}

\begin{theorem}\label{thm:correct}
SRC solves dispersion from every initial configuration $\mathcal{C}_0$ on the oriented square grid for all $k\le n$.
\end{theorem}

\begin{proof}
By Lemma~\ref{lem:noconflict}, no two settled agents ever share a node, so any all-settled configuration is dispersed.  It remains to show that every agent terminates settled.

Termination follows from Theorem~\ref{thm:main}.  For settling, consider any active agent.  If its subgroup path never meets a higher label, the path is clean and Lemma~\ref{lem:threecycles} either settles it or moves it to a higher level after a true doubling at $C$.  If it meets a higher moving group, it transfers immediately and, when the higher label is in \textsc{Reverse} or \textsc{Collect}, follows that phase.  If it meets a higher marker, Section~\ref{subsec:marker} sends it to the predicted subgroup head, the predicted reverse/collect phase, or the common destination $C$; missed edge crossings and wrong-path misses are absorbed by the $2\gamma$ wait/recovery rule.  If it remains in the lower label after other agents transfer upward, only its own contacted path is losing, and Lemma~\ref{lem:threecycles} gives finitely many further remnant cycles with no promotion from the smaller base.  Since levels are bounded by Lemma~\ref{lem:doubling}, this process cannot continue forever.  Eventually every agent is passive-settled, or is collected by a same- or higher-labelled \textsc{Collect} and then settled by that higher episode.  By Lemma~\ref{lem:noconflict}, the final settled configuration is dispersed, and every agent terminates. 
\end{proof}

\section{Discussion and Extensions}\label{sec:disc}

\paragraph{\textbf{Infinite grids.}}
It is worth highlighting that SRC does not require knowledge of the environment. 
It never reads $n$: the only possible use of grid size is the boundary-clipping test of Section~\ref{subsec:sweep}, resolved locally by checking whether the required outgoing port exists.  If it does, the
sweep proceeds; if not, the group has met an edge and reflects its square
back inside by the bounce of Section~\ref{subsec:boundary}. 
The grid enters at just one point: when an agent checks whether the port it needs exists.
Everything else a group does --- settling, merging, collecting and levelling up --- remains the same on infinite grids.

On such a grid there is no boundary, so every port an agent asks for is present, and no bounce occurs. A group then just sweeps, settles, and merges the way it does inside a finite grid, and Theorem~\ref{thm:main} still holds in the absence of bounce, so the running time remains the same, i.e., $O(\sqrt{k})$. Some of the earlier grid algorithms~\cite{BKM24, BKM25, KMS20} use the boundary as a place for the agents to meet and get organised, so when there is no boundary, their first step is no longer defined or costly in terms of rounds. We believe our ``beyond boundary'' dispersing agents would be independently useful for the more symmetric graph structures.

\paragraph{\textbf{Relative-coordinate memory.}}
All sweep coordinates are stored relative to the anchor $A$, which is re-anchored to the agent's current node at the start of every phase.  The largest recorded displacement is therefore the sweep extent $\le 2(\gamma-1)=O(\sqrt{k})$, fitting in $O(\log k)$ bits.  This matches the $\Omega(\log k)$ memory lower bound needed just to give $k$ agents distinct identifiers, so SRC is memory-optimal as well as time-optimal.

\section{Concluding Remark and Future Works}\label{sec:conclusion}
 We presented \textsc{Spread--Reverse--Collect}, a deterministic synchronous algorithm that disperses any $k\le n$ agents on an oriented square grid in $O(\sqrt{k})$ rounds using $O(\log k)$ bits per agent. This matches the $\Omega(\sqrt{k})$ time and $\Omega(\log k)$ memory lower bounds in the local communication model, closing the gap left by~\cite{KMS20}. SRC relies on one geometric primitive---the \sharpS-shaped parallel sweep---and one rule: the higher label wins. Its running time follows from the hierarchical merging-and-doubling argument: a label's certificate at least doubles at every level, contacted losing subgroup paths remain non-promoting, and the total work is bounded by $\sum_{\ell\le\log_2 k}O(2^{\ell/2})=O(\sqrt{k})$.

 A natural next step is dispersion on a rectangular $W\times H$ grid with $W\le H$ (w.l.o.g).
 The square assumption enters the boundary rule: on a square grid, $\gamma>\Side$ means that the intended sweep exceeds both dimensions. On a rectangle, a sweep can exceed the short dimension while remaining much smaller than the long one, so the same bounce and whole-grid cases do not cover that situation. The time (round) complexity also changes in a thin grid. From a corner, at most $W(t+1)$ nodes are reachable in $t$ rounds, so dispersing $k$ agents requires $t\ge\lceil k/W\rceil-1$. In particular, a line can require time linear in $k$, even when the total grid is much larger.
 Thus $\sqrt{k}$ is the right target only when the grid is wide enough, and a thin grid requires a different, dimension-dependent bound.
A strategy similar to SRC might still work, but the main task is to replace the square sweep by a shape that adapts to the dimensions of the grid. We leave this as an open problem.

Several other directions remain open: extending SRC to unoriented square grids, fault-tolerant settings (including crash and Byzantine faults), and asynchronous schedulers while preserving the $O(\sqrt{k})$ running time are natural next steps. It would also be interesting to determine whether the same  merging-and-doubling framework applies to other graph classes with sublinear diameter (with respect to the graph size).

\bibliographystyle{plain}
\bibliography{biblio_cleaned_updated}

\end{document}